%% file: ISIS_data_driven_MAIN.tex
\documentclass{sig-alternate}
\usepackage{url}
\usepackage{amsmath}
\usepackage{amssymb}
\usepackage{algorithmic}
\usepackage{algorithm}
\usepackage{graphicx}

\newtheorem{proposition}{Proposition}[section]
\newtheorem{definition}{Definition}[section]

\def\eavg{\epsilon_{avg}}
\def\emin{\epsilon_{min}}
\def\efrac{\epsilon_{frac}}
\def\prn{\textsf{)}}
\def\Th{\textit{Th}}
\def\cala{\mathcal{A}}
\def\tm{\textit{t}_{\max}}
\def\R{\textbf{R}}

\newfont{\mycrnotice}{ptmr8t at 7pt}
\newfont{\myconfname}{ptmri8t at 7pt}
\let\crnotice\mycrnotice%
\let\confname\myconfname%

\usepackage[
  pass,% keep layout unchanged 
]{geometry}

\begin{document}
%\conferenceinfo{WOODSTOCK}{'97 El Paso, Texas USA}

\title{Mining for Causal Relationships: A Data-Driven Study of the Islamic State}
\numberofauthors{6} %  in this sample file, there are a *total*
\author{
%\alignauthor
Andrew Stanton, Amanda Thart, Ashish Jain, Priyank Vyas, Arpan Chatterjee, Paulo Shakarian\thanks{Paulo Shakarian is also affiliated with ASU Center on the Future of War.}\\
       \affaddr{Arizona State University}\\
			 \affaddr{Tempe, AZ 85287}\\
       \email{\{dstanto2, althart, ashish.jain.1, pvyas1, achatt14, shak\} @asu.edu}	
}

\maketitle
\begin{abstract}
The Islamic State of Iraq and al-Sham (ISIS) is a dominant insurgent group operating in Iraq and Syria that rose to prominence when it took over Mosul in June, 2014.  In this paper, we present a data-driven approach to analyzing this group using a dataset consisting of $2200$ incidents of military activity surrounding ISIS and the forces that oppose it (including Iraqi, Syrian, and the American-led coalition).  We combine ideas from logic programming and causal reasoning to mine for association rules for which we present evidence of causality.  We present relationships that link ISIS vehicle-bourne improvised explosive device (VBIED) activity in Syria with military operations in Iraq, coalition air strikes, and ISIS IED activity, as well as rules that may serve as indicators of spikes in indirect fire, suicide attacks, and arrests.
\end{abstract}
\input{introduction}

\input{apt}

\input{event_data}

\input{results}
\input{related_work}
\input{conclusion}

\end{document}

%% file: introduction.tex
\section{Introduction}
Since its rise to prominence in Iraq and Syria in June, 2014, The Islamic State of Iraq and al-Sham (ISIS) -- also known as The Islamic State of Iraq and the Levant (ISIL) or simply the Islamic State, has controlled numerous cities in Sunni-dominated parts of Iraq and Syria.  ISIS has displayed a high level of sophistication and discipline in its military operations in comparison to similar insurgent groups, which perhaps may be the source of its success.  We have meticulously encoded and recorded $2200$ incidents of military activity conducted by ISIS and forces that oppose it (including Iraqi, Syrian, and the American-led coalition) in a relational database.  Our goal was to achieve a better understanding of how this group operates - which can lead to new strategies for mitigating ISIS's operations.  Specifically, we sought to analyze the behavior of ISIS using concepts from logic programming (in particular APT logic~\cite{apt11,apt12}) and causal reasoning~\cite{kb09,suppes70}.  By combining ideas from these fields, we have been able to conduct a thorough search for rules whose precondition consists of multiple atomic propositions, and we provide evidence of causality by comparing rules with the same consequence.  So, in addition to considering the probability of a rule ($p$), we also study a measure of its causality denoted $\eavg$ (previously introduced in~\cite{kb09}) -- which, informally, can be thought of as the average increase in probability a rule's precondition provides when considered with each of the comparable rules.  Using this approach, we have found interesting relationships - consider the following:
\begin{itemize}
\item Weeks where ISIS conducts infantry operations in Iraq that are accompanied by indirect fire are indicative of vehicle-bourne improvised explosive device (VBIED) operations in Syria in the following week ($p=1.0,$ $\eavg=0.92$).
\item Weeks in which ISIS conducts operations in Tikrit and conducts a significant number of executions are followed by a large spike in improvised explosive device (IED) usage in Iraq and Syria combined ($p=1.0, \eavg=0.97$)
\item Air strikes by the Syrian government are followed by mass arrests by ISIS in the following week ($\eavg=0.91, p=0.67$), and such massive arrests were \textit{always} proceeded by Syrian air strikes in our dataset.
\item In the week after coalition air strikes are conducted against Mosul while ISIS is conducting operations in Al-Anbar province, ISIS greatly increases its IED activity in Iraq ($p=0.67, \eavg=0.97$).  However, if there are also significant ISIS operations occurring in Syria, the increase in IED usage experienced there is ($p=0.67, \eavg=0.79$).
\end{itemize}
Our findings have also led us to several interesting theories about ISIS behavior that we have developed as a result of this data mining effort.  These include the following:
\begin{itemize}
\item ISIS may employ suicide VBIED operations in Baghdad prior to significant infantry operations in other locations to prevent the deployment of Iraqi army/police reinforcements.
\item ISIS tends to leverage indirect fire (IDF) as a precursor to infantry operations - more in keeping with a traditional military force as opposed to a primary use of IDF for harassment purposes (as was typically seen by insurgent groups during Operation Iraqi Freedom).
\item As we found relationships between coalition air operations and an increase in ISIS usage of IEDs - and not other, larger weapons systems (i.e. VBIEDs), this may indicate that ISIS resorts to more distributed, insurgent-style tactics in the aftermath of such operations.
\end{itemize}
To our knowledge, this study represents both the first publicly-available, data-driven study of ISIS as well as the first combination of APT logic with the causality ideas of \cite{kb09}.  The rest of the paper is organized as follows.  In Section~\ref{appr}, we describe our approach and recall key concepts from APT logic and causal reasoning.  This is followed by a description of our corpus of military events surrounding the actions of ISIS from June-December, 2014 in Section~\ref{dataset}.  In Section~\ref{resSec}, we describe our implementation and discuss our results.  Finally, related work is reviewed in Section~\ref{rwSec}.

%% file: apt.tex
\section{Technical Approach}
\label{appr}
\subsection{APT Logic}
We now present a subset of our previously-introduced APT logic \cite{apt11,apt12}.  Our focus here is on the syntax of the language utilized with an alternate semantics that is based on the rule-learning approach described in \cite{apt11}.  This is due to the fact that we are only concerned with learning the rules in this paper and determining their level of causality -- we leave problems relating to deduction (i.e. the entailment problem studied in \cite{apt12}) to  future work.  APT logic considers a two-part semantic structure: threads (sequences of events) and interpretations (probability distributions over such sequences).  Initially, we focus on a semantics restricted to threads.

We assume the existence of a set of ground atoms, $\cala$ which in this case correspond with events over time and use the symbol $n$ to denote the size of this set.  We shall partition the ground atoms into two subsets: action atoms $\cala_{act}$ and environmental atoms $\cala_{env}$. Action atoms describe actions of certain actors, and environmental atoms describe aspects of the environment.  We will use $n_{act},n_{env}$ to denote the quantity of ground atoms in $\cala_{act},\cala_{env}$, respectively.  We can connect atoms using $\neg,\wedge,\vee$ to create formulas in the normal way.  A \textit{world} is a subset of atoms that are considered to be true.  A \textit{thread} is a series of worlds, each corresponding to a discrete time-point, which are represented by simple natural numbers in the range $1,\ldots,\tm$.  We shall use $\Th[t]$ to refer to the world specified by $\Th$ at time $t$.  A thread ($\Th$) at time $t$ \textit{satisfies} a formula $f$  (denoted $\Th[t]\models f$) by the following recursive definition:
\begin{itemize}
\item{For some $a \in \cala$, $\Th[t]\models a$ iff $a \in \Th[t]$}
\item{For $f = \neg f'$, $\Th[t]\models f$ iff $\Th[t]\models f'$ is false}
\item{For $f=f'\wedge f''$, $\Th[t]\models f$ iff $\Th[t]\models f'$ and $\Th[t] \models f''$}
\item{For $f=f' \vee f''$, $\Th[t] \models f$ iff $\Th[t]\models f'$ or $\Th[t] \models f''$}
\end{itemize}

We also note that based on how often a given formula occurs within a thread, we can obtain a prior probability for that formula.  For instance, consider formula $g$. We can compute its prior probability (w.r.t. thread $\Th$) as the fraction of time points where $g$ is satisfied:
\begin{equation}
\label{rhComp}
\rho = \dfrac{|\{t | \Th[t]\models g\}|}{\tm}
\end{equation}
Note that we use the notation $|\cdot|$ for the cardinality of a set. Our goal in this section is for some event $g$ to identify previously occurring event $c$ such that the conditional probability of $g$ occurring given $c$ is greater than $p_g$.  In this paper, we shall only look at finding APT rules where $c$ occurs in the time period right before $g$.  The case where multiple time periods elapse between $c$ and $g$ is left to future work.  Hence, we introduce an APT rule of the following form: $c \leadsto_p g$, which intuitively means that ``$c$ is followed by $g$ in one time-step with probability $p$.''  We shall refer to $c$ as the precondition and $g$ as the consequence.  We say thread $\Th \models c \leadsto_p g$ iff:
\begin{equation}
\label{pComp}
p= \frac{|\{t \textit{ s.t. } \Th[t]\models c \textit{ and } \Th[t+1] \models g \}|}{|\{t \textit{ s.t. }\Th[t] \models c\}|}
\end{equation}
This alternate definition of semantics is not new: it is a variant of rule satisfaction with the existential frequency function of \cite{apt11} and also of ``trace semantics'' used for a PCTL variant in \cite{kb11}.  In this paper, we are focused on learning rules where $c$ is a conjunction of positive atoms and $g$ is a single positive atom.  The number of atoms in the conjunction $c$ is referred to as the ``dimension'' of the rule $c\leadsto_p g$.  We will also be concerned with two other measures of a given rule: the ``negative probability'' (used to describe the probabilistic rules of \cite{let13}) and the notion of support (a standard concept in association rule learning).  First, we define ``negative probability'' -- denoted $p^{*}$ -- for a given rule of the aforementioned format.
\begin{equation}
\label{psComp}
p^{*}= \frac{|\{t \textit{ s.t. } \Th[t]\models g \textit{ and } \Th[t-1] \not\models c \}|}{|\{t \textit{ s.t. }\Th[t] \models g\}|}
\end{equation}
Simply put, $p^{*}$ is the number of times that the consequence of a rule occurs without the head occurring prior.  We can think of $p$ as a measure of precision of a given rule while $p^{*}$ is more akin to a measure of recall.  Next, we formally define the notion of support ($s$) as follows.
\begin{equation}
\label{sComp}
s= |\{t \textit{ s.t. }\Th[t] \models c\}|
\end{equation}
For a given rule $r$, we shall use the notation $p_r, p^{*}_r, s_r$, respectively.  We will also use the notation $\rho_r$ to denote the prior probability of the consequence.  When it is obvious from context which rule we are referring to (as in the experimental results section), we will drop the subscript.

\subsection{Determining Causality}
\label{deteCauseSec}
Next, we describe how for a given action, $g$, we identify potential causes based on a set of rules for which $g$ is the consequence.  In considering possible causes, we must first identify a set of \textit{prima facie} causes for $g$ \cite{kb09,suppes70}.  We present the definition in terms of APT logic below.

\begin{definition}[Prima Facie Cause~\cite{kb09,suppes70}]
\label{pfcDef}
Given\\ APT logic formulas $c,g$ and thread $\Th$, we say $c$ is a \textbf{prima facie} cause for $g$ w.r.t. $\Th$ if:
	\begin{enumerate}
	\item There exists time $t$ such that $\Th[t] \models g$ ($g$ occurs with a probability greater than zero)
	\item There exists $t,t'$ where $t<t'$ such that $\Th[t]\models c$ and $\Th[t']\models g$ ($c$ occurs before $g$)
	\item For $r \equiv c\leadsto_p g$ where $\Th\models r$, we have $p> \rho_r$ (the probability of the consequence occurring after the precondition is greater than the prior probability of the consequence)
	\end{enumerate}
	\end{definition}

As we will assume the existence of a single thread, $\Th$ (which is our historical corpus of data), we will often use the language ``prima facie causal rule'' or ``PF-rule'' to describe a rule $c \leadsto_p g$ where $c$ is a prima facie cause for $g$ w.r.t. $\Th$.\smallskip

Next, we adopt the method of \cite{kb09} to determine if an APT rule is causal.  First, in determining if a given PF-rule is causal, we must consider other, related PF-rules.  Intuitively, a given rule $r \equiv c \leadsto_{p} g$ explains why some instances of $g$ occur within a thread.  Another rule $r' \equiv c' \leadsto_{p'} g$ is related if it also explains why some of those same instances occur.  Formally, we say $r$ and $r'$ are related if there exists $t$ such that $\Th[t]\models c\wedge c'$ and $\Th[t+1] \models g$.  Hence, we will define related PF-rules (for a given rule $r \equiv c \leadsto_p g$) as follows.
\[
\R(r) = \{r' \textit{ s.t. } r' \not\equiv r \textit{ and } r,r' \textit{ are related}\}
\]
Next, we will look at how to compare two related rules.  The key purpose from \cite{kb09} behind doing so is to study how the probability of the rule changes in cases where both preconditions co-occur in comparison to the probability where exactly one of the preconditions occurs.  So, given rules $r,r'$ as defined above, we have the following notation: $p_{r,r'}$ and $p_{\neg r,r'}$, which are defined as the point probabilities that will cause the following two rules to be satisfied by $\Th$.
\begin{eqnarray*}
c \wedge c' \leadsto_{p_{r,r'}}g\\
\neg c \wedge c' \leadsto_{p_{\neg r, r'}} g
\end{eqnarray*}
So, $p_{r,r'}$ is the probability that $g$ occurs given both preconditions, and $p_{\neg r, r'}$ is the probability that $g$ occurs given just the precondition of the second rule and not of the first.  The idea is that if $p_{r,r'}-p_{\neg r, r'} >0$, then there is something about the precondition of $r$ that causes $g$ to occur that is not present in $r'$.  Following directly from \cite{kb09}, we measure the average of this quantity to determine how causal a given rule is as follows:
\begin{eqnarray*}
\label{eavgEqn}
\eavg &=& \dfrac{\sum_{r' \in \R(r)}p_{r,r'}-p_{\neg r, r'}}{|\R(r)|}
\end{eqnarray*}

Intuitively, $\eavg(r)$ measures the degree of causality exhibited by rule $r$.  Additionally, using this same intuition, we find it useful to include a few other related measures when examining causality.  First, we define $\emin$, defined below:
\begin{eqnarray*}
\emin &=& \min_{r' \in \R(r)}(p_{r,r'}-p_{\neg r, r'})
\end{eqnarray*}
This tells us the ``least causal'' comparison of $r$ with all related rules.  Another measure we will use is $\efrac$, defined as follows:
\begin{eqnarray*}
\efrac(r) &=& \dfrac{|\{r' \textit{ s.t. } (p_{r,r'}-p_{\neg r, r'})\geq 0 \}|}{|\R(r)|}
\end{eqnarray*}
This provides a fraction of the related rules whose probability remains the same or decreases if the precondition of $r$ is not present.  Hence, a number closer to $1$ likely indicates that $r$ is more causal.  By using multiple causality measurements, we can have a better determination of more significant causality relationships.

\subsection{Algorithms}

Next, given a thread, we provide a variant of the APT-Extract algorithm~\cite{apt11} that we call PF-Rule-Extract.  Essentially, this algorithm generates all possible $\cala_{env}$ up to a certain size (specified by the argument $MaxDim$) and then searches for correlation.  It returns rules whose precondition occurs a specified number of times (a lower bound on support for a rule) - specified with the argument $SuppLB$.  We make several modifications specific to our application to support reasoning. 
\begin{itemize}
\item We reduce the run-time of the algorithm considerably over APT-Extract.  As APT-Extract explores all possible combinations  of atoms up to size $MaxDim$, it runs in time $O({n_{env} \choose MaxDim})$.  However, PF-Rule-Extract only examines combinations of atoms that occur in a given time period -- hence, reducing run-time to  $O({\max_t(n_t) \choose MaxDim})$ where $n_t$ is the number of atoms in $\cala_{env}$ true at time $t$.  We find in practice that $\max_t(n_t) << n_{env}$.
\item We further reduce the run-time by not considering elements of $\cala_{env}$ that occur less than the lower bound on support - as they would never appear in a rule.
\item We guarantee that all rules returned by PF-Rule-Extract are PF-rules.
\end{itemize}

\algsetup{indent=1em}
	\begin{algorithm}[h!]
		\caption{ \textsf{PF-Rule-Extract}}
		\label{csd_alg}
		\begin{algorithmic}[1]
		\REQUIRE Thread $\Th$, sets of atoms $\cala_{act},\cala_{env}$, positive natural numbers $MaxDim, SuppLB$, real number $minProb$\\
		\ENSURE Set of rules $ \textbf{R}  $
		\medskip
		\STATE{Set $\textbf{R}=\emptyset$}
		\FOR{$g \in \{a \in \cala_{act} \textit{ s.t. } \exists t \textit{ where } \Th[t] \models a\}$}
			\STATE{$preCond(g) = \emptyset$}
			\FOR{$t \in \{1,\ldots,\tm \textit{ s.t. } \Th[t+1]\models g\}$} 
				\STATE{Let $X$ be the set of all combinations of size $MaxDim$ (or less) of elements in $\Th[t]\cap \cala_{env}$ that occur at least $SuppLB$ times.}
				\STATE{$preCond(g) = preCond(g)\cup X$}% \cup \{a \in \cala_{env} \textit{ s.t. there is only one }t\textit{ where }a \in \Th[t]\}$}
			\ENDFOR
			\FOR{$c \in preCond(g)$}
				\STATE{Compute $\rho, p, s$,  as per Equations~\ref{rhComp},\ref{pComp}, and \ref{sComp}} %p^*, psComp
					\IF{$(s \geq SuppLB) \wedge (p > \rho) \wedge (p\geq minProb)$}
							\STATE{$\textbf{R}=\textbf{R}\cup \{r\}$}
					\ENDIF
			\ENDFOR
		\ENDFOR
		\RETURN{Set $\textbf{R}$}%_g \textit{ s.t. } g \in \cala_{act}\}$}
			\end{algorithmic}
\end{algorithm}

Note that APT-Extract allows for multiple time periods between preconditions and consequences -- PF-Rule-Extract can also be easily modified to find rules of this sort through a simple modification of line 4.  Next, we provide some formal results to show that PF-Rule-Extract finds rules that meet the requirements of Definition~\ref{pfcDef} and show that it explores all possible combinations up to size $MaxDim$ that are supported by the data that meet the requirement for minimum support even with our efficiency improvements.

\begin{proposition}
\label{satProp}
Given thread $\Th$ as input, PF-Rule-Extract produces a set of rules $\textbf{R}$ such that each $r \in R$ is a prima facie causal rule.
\end{proposition}
\begin{proof}
The first requirement for a prima facie cause is met by line 2 of PF-Rule-Extract, as the algorithm only considers consequences that have occurred at least once in $\Th$ -- hence, they have a prior probability greater than zero.  The second requirement is met by line 4, as we only consider atoms that occurred immediately before the consequence when constructing the precondition.  Finally, the third condition is met by the if statement at line 10 which will only select rules whose probability is greater than the prior probability of the consequence. 
\end{proof}

\begin{proposition}
\label{allProp}
PF-Rule-Extract finds all PF-rules whose precondition is a conjunction of size $MaxDim$ or a size containing fewer atoms and where both the precondition occurs at least $SuppLB$ times in $\Th$ and the probability is greater than $minProb$.
\end{proposition}
\begin{proof}
By Proposition~\ref{satProp}, every consequence that could be used in a PF-rule is considered, so we need only to consider the precondition.  Suppose that there is a PF-rule whose precondition occurs at least $SuppLB$ times that has a size of $MaxDim$ or smaller.  Hence, the precondition is comprised of $m \leq MaxDim$ atoms: $a_1,\ldots,a_m$.  Clearly, as this precondition occurs at least $SuppLB$ times in $\Th$, each of $a_1,\ldots,a_m$ must also occur $SuppLB$ times by line 5 - so that the precondition must be considered at that point.  The only condition that would prevent the rule from being returned is line 10, but again, clearly these are met by the statement of the proposition.  Hence, we have a contradiction.
\end{proof}

Once we have identified the PF-rules using PF-Rule-Extract, we must then compare related rules with each other.  The  algorithm PF-Rule-Compare returns the top $k$ rules for each consequence.  Again, here we take advantage of our specific application to reduce the run-time of this comparison.  In particular, for a given rule, we need not compare it to all the rules returned by PF-Rule-Extract; we only need to compare it to those that share the consequence.  By sorting the rules by consequence - a linear time operation (line 3), we are able to reduce the cost of the quadratic operation for the rule comparisons (a brute-force method would take $O(|\textbf{R}|^2)$ while this method requires $O((\max_g|\textbf{R}_g|)^2+|\textbf{R}|)$ -- and we have observed that $\max_g|\textbf{R}_g| << |\textbf{R}|$).

\algsetup{indent=1em}
	\begin{algorithm}[h!]
		\caption{ \textsf{PF-Rule-Compare}}
		\label{compar_alg}
		\begin{algorithmic}[1]
		\REQUIRE Set of rules $\textbf{R}$, natural number $k$ \\
		\ENSURE Set of rules $ \textbf{R}'  $
		\medskip
		\STATE{Set $\textbf{R}'=\emptyset$.}
		\FOR{$g \in \cala_{act}$}
			\STATE{Set $\textbf{R}_g = \{ c \leadsto_p g \in \textbf{R}\}$}
			\FOR{$r \in \textbf{R}_g$}
				\STATE{Calculate $\eavg$ for rule $r$ (Equation~\ref{eavgEqn} by comparing it to all related rules in $\textbf{R}_g\setminus\{r\}$})
			\ENDFOR
			\STATE{From the set $\textbf{R}_g$, add the top $k$ rules by $\eavg$ to the set $\textbf{R}'$}
		\ENDFOR
		\RETURN{Set $\textbf{R}'$}
			\end{algorithmic}
\end{algorithm}

%% file: event_data.tex
\section{ISIS Dataset}
\label{dataset}
We collected data on $2200$ military events that occurred from June 8th through December 31st, 2014 that involved ISIS and forces opposing ISIS.  Events were classified into one of $159$ event types - and these events were used as predicates for APT logic atoms.  We list a sampling of the predicates we used in Table~\ref{exAt}.  Many predicates are unary because they correspond with ISIS actions (i.e. \textsf{armedAtk}) while those dealing with operations by other actors (i.e. \textsf{airOp}) and predicates denoting spikes in activity (i.e. \textsf{VBIEDSpike}) accept more than one argument.  Nearly every predicate has one argument that corresponds to the location in which the associated event took place. See Figure \ref{fig:map} for a sampling of locations in the ISIS dataset.\smallskip

We obtained our data primarily from reports published by the Institute for the Study of War (ISW) \cite{isw14}, and we augmented it with other reputable sources including MapAction \cite{ma14i,ma14s}, Google Maps \cite{ggmps}, and Humanitarian Response \cite{hr14}.  We devised a code-book as well as coding standards for events, and one of our team members functioned as a quality-control to reduce errors in human coding.

\begin{figure}[htp!]
	\vspace{-1em}
	\caption{\textmd{Map illustrating cities associated with events in the ISIS dataset.}}
		\label{fig:map}
	\includegraphics[width=8.5cm]{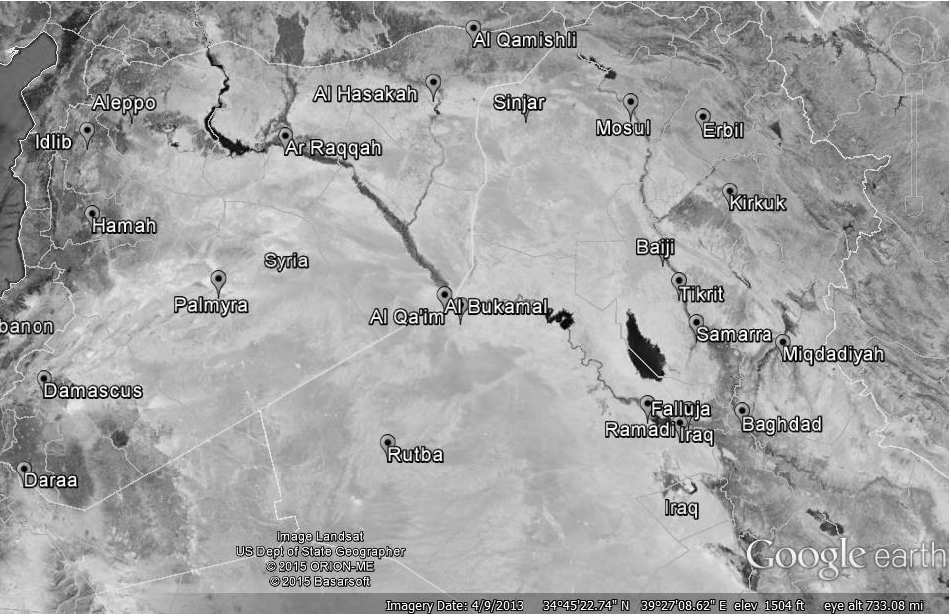}
\end{figure}

\begin{table}[h!]
	\caption{\textmd{A sampling of predicate symbols used to describe events in the ISIS Dataset}}
	\label{exAt}
	\centering
	\renewcommand{\arraystretch}{1.5}
	
	\begin{tabular}{|p{3cm}|p{5cm}|} 
		\hline
		{\bf Predicate} &  {\bf Intuition} \\ \hline 
		\textsf{airOp(}$X,Y$\textsf{)} & Actor $X$ (typically ``Coalition,'' ``Syrian Government,'' or ``U.S.'') conducts an air campaign against ISIS in the vicinity of city $Y$.\\ \hline
				\textsf{armedAtk(}$X$\textsf{)} & ISIS conducts an armed attack (a.k.a. infantry operation) in city $X$.\\ \hline
				\textsf{IED(}$X$\textsf{)} & ISIS conducts an attack using an improvised explosive device (IED) in city $X$.\\ \hline
		\textsf{indirectFire(}$X$\textsf{)} & ISIS conducts an indirect fire operation (i.e. mortars) near city $X$.\\ \hline
		\textsf{VBIED(}$X$\textsf{)} & ISIS conducts an attack using a vehicle-bourne improvised explosive device (VBIED or car-bomb) in city $X$.\\ \hline\hline
		\textsf{armedAtkSpike(}$X,Y$\textsf{)} & There is a spike in ISIS armed attacks in country $X$ that is $Y$ amount over the 4-month moving average ($Y$ is expressed in terms of moving standard deviations).\\ \hline
				\textsf{VBIEDSpike(}$X,Y$\textsf{)} & There is a spike in ISIS VBIEDs in country $X$ that is $Y$ amount over the 4-month moving average ($Y$ is expressed in terms of moving standard deviations).\\ \hline

				\hline
	\end{tabular}
	
\end{table}

%% file: results.tex
\section{Results and Discussion}
\label{resSec}
\subsection{Experimental Setup}
We implemented our PF-Rule-Extract and PF-Rule-Compare in Python 2.8x and ran it on a commodity machine equipped with an Intel Core i5 CPU (2.7 GHz) with 16 GB of RAM running Windows 7.  The time periods we utilized were weeks - hence, we had $30$ time periods in our thread.  We had $980$ distinct environmental atoms ($\cala_{env}$).  Our action atoms ($\cala_{act}$) corresponded to weeks where the number of incidents for certain activity rose to or past one or two moving standard deviations (denoted $1\times \sigma, 2\times \sigma$ respectively, computed based on the previous 4 weeks) above the four-week moving average (computed based on the previous 4 weeks).  We refer to these atoms as ``spikes'' and show some sample atoms denoting such spikes in Table~\ref{exAt}.  The spikes are designated for three locations: Iraq, Syria, or both theaters combined.  We also included these spikes in the set of environmental atoms as well ($\cala_{act}\subset \cala_{env}$).\smallskip\\

We set the parameters of PF-Rule-Extract as follows:\\ $MaxDim=3$, $SuppLB=3$, $minProb=0.5$.  For PF-Rule-Compare, we did not set a particular $k$ value. Instead, we used $\eavg$ to rank the rules by causality for a given consequence.  In this paper, we report the top causal rules (in terms of $\eavg$) for some of the consequences.  We also note that the number of related rules provides insight into that rule's significance - we discuss this quantity in our analysis.  We examine some of the top rules in terms of $\eavg$ and describe the military insights that they provide.  Our team includes a former military officer with over a decade of experience in military operations that includes two combat tours in Operation Iraqi Freedom.  In addition to $\eavg$, we also examined the rule's probability ($p$) - the fraction of times the consequence follows the precondition, the negative probability ($p^*$) - the fraction of times the consequence is not proceeded by the precondition, the prior probability of the consequence ($\rho$), and the additional causal measures introduced in this paper - $\emin, \efrac$.  The relationship specified in the rule can be considered more significant if $p>>\rho$, $p^*$ is closer to $0$, $\emin \geq 0$, and $\eavg,\efrac$ are close to $1$.  These measures are discussed in more detail in Section~\ref{appr}.

\subsection{Algorithm Efficiency}
In Section~\ref{appr}, we described some performance improvements that we utilized in PF-Rule-Extract (which is based on APT-Extract~\cite{apt11}) that are specific to this application.  The first performance improvement was the reduction in the number of atoms used to generate the preconditions (which were conjunctions of atoms of up to size $3$ in our experiments).  First, we limited the number of atoms to be used in such combinations based on the weeks in which they occurred.  Even though we had $980$ environmental atoms, no more than $93$ occurred during any given week (this is the value $\max_t(n_t)$ from Section~\ref{appr}).  Further, by eliminating atoms that occurred less than the $SuppLB$ from consideration, this lowered it further to $49$ atoms per week at the most.  This directly leads to fewer combinations of atoms generated for the precondition.  A comparison is shown in Table~\ref{comboTable}.  Note that the values for APT-Extract are exact, while the remaining are upper bounds.  Again, this substantial, multiple-order-of-magnitude savings in the number of rules explored is a result of the relative sparseness of our dataset and the fact that our preconditions consisted of conjunctions of positive atoms.  This efficiency is primarily what enabled us to find and compare rules in a matter of minutes on a commodity system.

\begin{table}[h!]
	\caption{\textmd{Improvement in Algorithm Efficiency}}
	\label{comboTable}
	\centering

	\begin{tabular}{|l|l|l|} 
		\hline
		{\bf Technique} &  {\bf Atoms}&  {\bf Combinations} \\ 
		&&{\bf Explored}\\\hline 
		APT-Extract~\cite{apt11} & $980$& $182,122,025$\\\hline
		$\max_t(n_t)$ & $93$ &$8,853,042$\\\hline
		$\max_t(n_t)$ and consider& $49$ & $1,296,834$\\
		only atoms that occur more&&\\
		than $SuppLB$ times && \\
		%times &&\\
		\hline
		\hline
			\end{tabular}
\end{table}

\subsection{ISIS Military Tactics}
In this section, we investigate rules that provide insight into ISIS's military tactics -- in particular, we found interesting and potentially casual relationships that provide insight into their infantry operations, use of terror tactics (i.e. VBIEDs), and decisions to employ roadside bombs and to launch suicide operations.

\begin{table}[h!]
	\caption{\textmd{Causal Rules for Spikes in Armed Attacks (Iraq and Syria Combined)}}
	\label{armedAttackRules}
	\centering
	\renewcommand{\arraystretch}{1.5}
	
	\begin{tabular}{|c|c|c|c|c|} 
		\hline
		\textbf{No.} &{\bf Precondition} &  {\bf $\eavg$} & {\bf $p$} & {\bf $p^*$} \\ \hline\hline
		\textbf{1.}&$\textsf{indirectFire(}Baiji\prn\wedge$ & $0.81$ & $0.67$ & $0.50$\\ 
	  &$ \textsf{armedAtk(}Balad\prn  $ &&& \\ \hline
		\textbf{2.}&$\textsf{indirectFire(}Baiji\prn \wedge $ & $0.81$ & $0.67$ & $0.50$  \\ 
			&$ \textsf{armedAtk(}Balad\prn\wedge  $ &&& \\
			&$ \textsf{VBIED(}Baghdad\prn  $ &&& \\ \hline
		\hline
	\end{tabular}
\end{table}

\noindent\textbf{Armed Attacks.}  The prior probability of large spikes ($2 \times \sigma$) (see figure \ref{fig:chart1}) for ISIS armed attacks for Iraq and Syria was $0.154$. However, for our two most causal rules for this spike (Table~\ref{armedAttackRules}), we derived this probability as $0.67$. Such spikes likely indicate major infantry operations by the Islamic State.  Rule 1 states that indirect fire at Baiji (attributed to ISIS) along with an armed attack in Balad leads to a spike in armed attacks by the group in the next week while rule 2 mirrors rule 1 but adds VBIED activity in Baghdad as part of the precondition.  We note the relatively high $p^*$ ($0.5$ in this case), which indicates that each spike in armed attacks of this type are not necessarily proceeded by this precondition, despite the relatively high value for $\eavg$, hinting at causality (each of these rules was compared with $265$ related rules and had $\efrac=1$ and $\emin=0$ -- showing little indication of a related rule being more causal).  We believe that the strong causality and the high value for $p^*$ indicate that these rules may well be ``token causes'' -- causes for a specific event - in this case, we think it is likely ISIS offensive operations in Baiji.  This makes sense, as a common military tactic is to prepare the battlefield with indirect fire (which it seems ISIS did in the prior week).  It is unclear if the VBIED incidents in Baghdad are related due to the co-occurrence.  That said, it is notable that VBIED operations are often used as ``terror'' tactics as opposed to part of a sustained operation.  If so, this may have been part of a preparatory phase (that included indirect fire in Baiji), and the purpose of the VBIED events in Baghdad was to prevent additional Iraqi Security Force deployment from Baghdad to Baiji.  We note in one case supporting this rule (on July 25th, 2014) that ISIS also conducted IED attacks on power-lines that support Baghdad - which may have also been designed to hinder deployment of reinforcements.  Here, it is also important to note the ongoing Infantry operations in Balad (specified in the precondition) - which would consume ISIS resources and perhaps make it more difficult to respond to further deployment of government security forces.  Though this is likely a token cause, this may be indicative of ISIS tactics when preparing to concentrate force on certain objectives (in this case Baiji) while maintaining ongoing operations (Balad), as a spike in armed attacks likely indicates a surge of light infantry-style soldiers into the area (a manpower-intensive operation). \smallskip

\begin{figure}[h!]
\caption{ISIS Spikes in Armed Attacks in Iraq and Syria per Week}
	\label{fig:chart1}
\includegraphics[width=8.5cm]{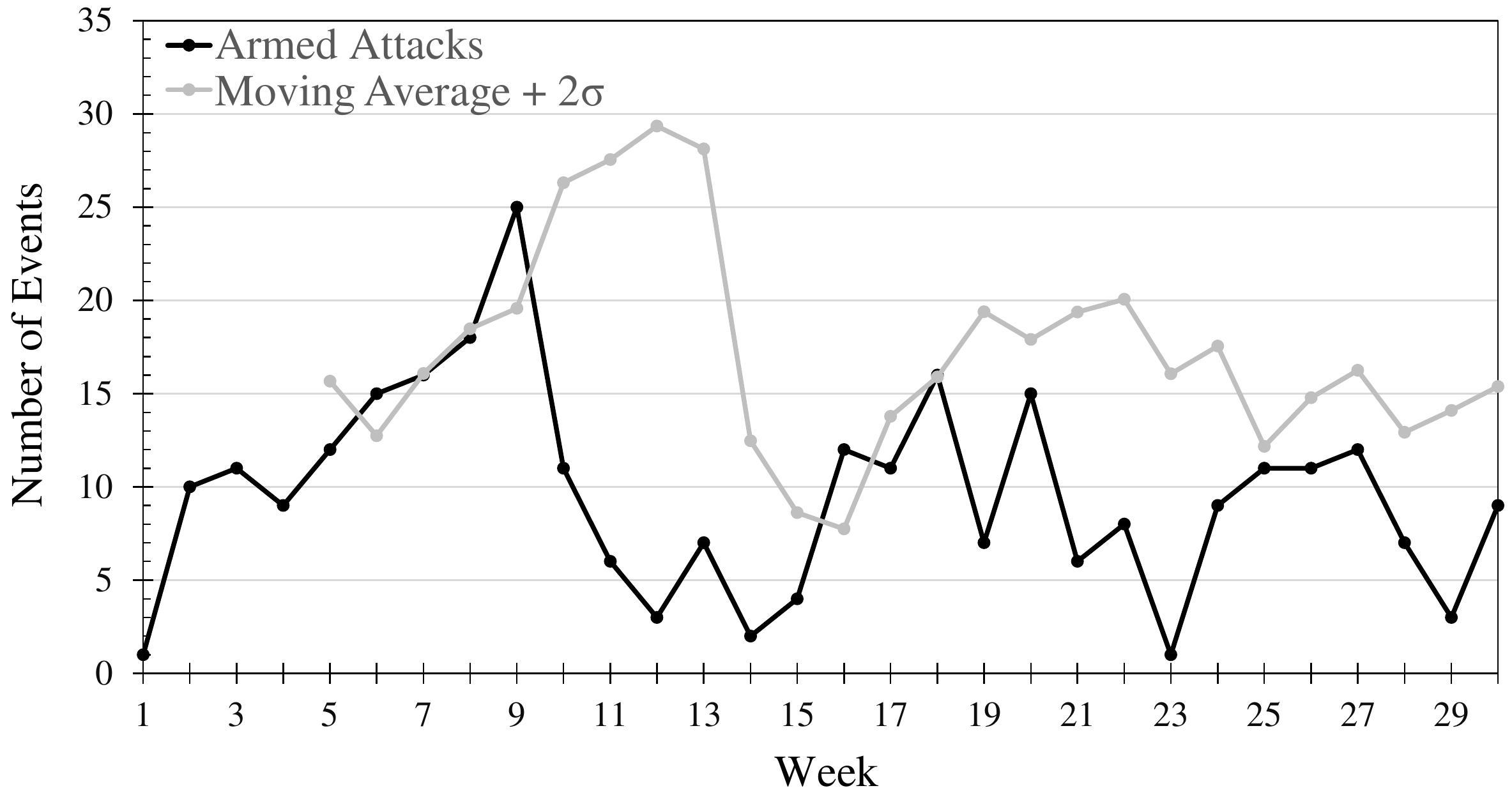}
\end{figure}

\noindent\textbf{Spikes in VBIED Incidents.}  VBIED incidents have been a common terror tactic used by religious Sunni extremist insurgent groups in Iraq in the past -- including Al Qaeda in Iraq, Ansar al Sunnah, Ansar al Islam, and now ISIS.  We found two relationships that are potentially causal for spikes in VBIED activity in Iraq and Syria combined, shown in Table~\ref{vbiedSpikeRules}.

\begin{table}[h!]
	\caption{\textmd{Causal Rules for Spikes in VBIED Incidents in Iraq and Syria Combined}}
	\label{vbiedSpikeRules}
	\centering
	\renewcommand{\arraystretch}{1.5}
	
	\begin{tabular}{|c|c|c|c|c|} 
		\hline
		\textbf{No.} &{\bf Precondition} &  {\bf $\eavg$} & {\bf $p$} & {\bf $p^*$} \\ \hline\hline
		\textbf{3.}&$\textsf{armedAtk(}Balad\prn \wedge $ & $0.95$ & $1.00$ & $0.25$\\ 
		& $\textsf{indirectFire(}Baiji\prn$&&&\\ \hline
		\hline
	\end{tabular}
\end{table}

Rule 3 states that if infantry operations in Balad accompanied by indirect fire operations in Baiji occur, we should expect a major ($2\times \sigma$) spike in VBIED activity (see figure \ref{fig:chart2}) by ISIS (Iraq and Syria combined).  We believe this rule provides further evidence of the use of VBIEDs to pull security forces away from other parts of the operational theater.  It is interesting to note that the negative probability is only $0.25$ - which means that most of the VBIED spikes we observed were related to operations in Balad and Baiji.  This highlights the strategic importance of these two cities to ISIS: Baiji is home to a major oil refinery while Balad is near a major Iraqi air base.  We also found solid evidence of causality for this rule - based on $209$ related rules, we found $\emin=0.5$ which means adding a second precondition to this rule increases the probability of the second rule by at least $0.5$\smallskip

\begin{figure}[h!]
\caption{Spikes in ISIS VBIED Activity in Iraq and Syria per Week}
	\label{fig:chart2}
\includegraphics[width=8.5cm]{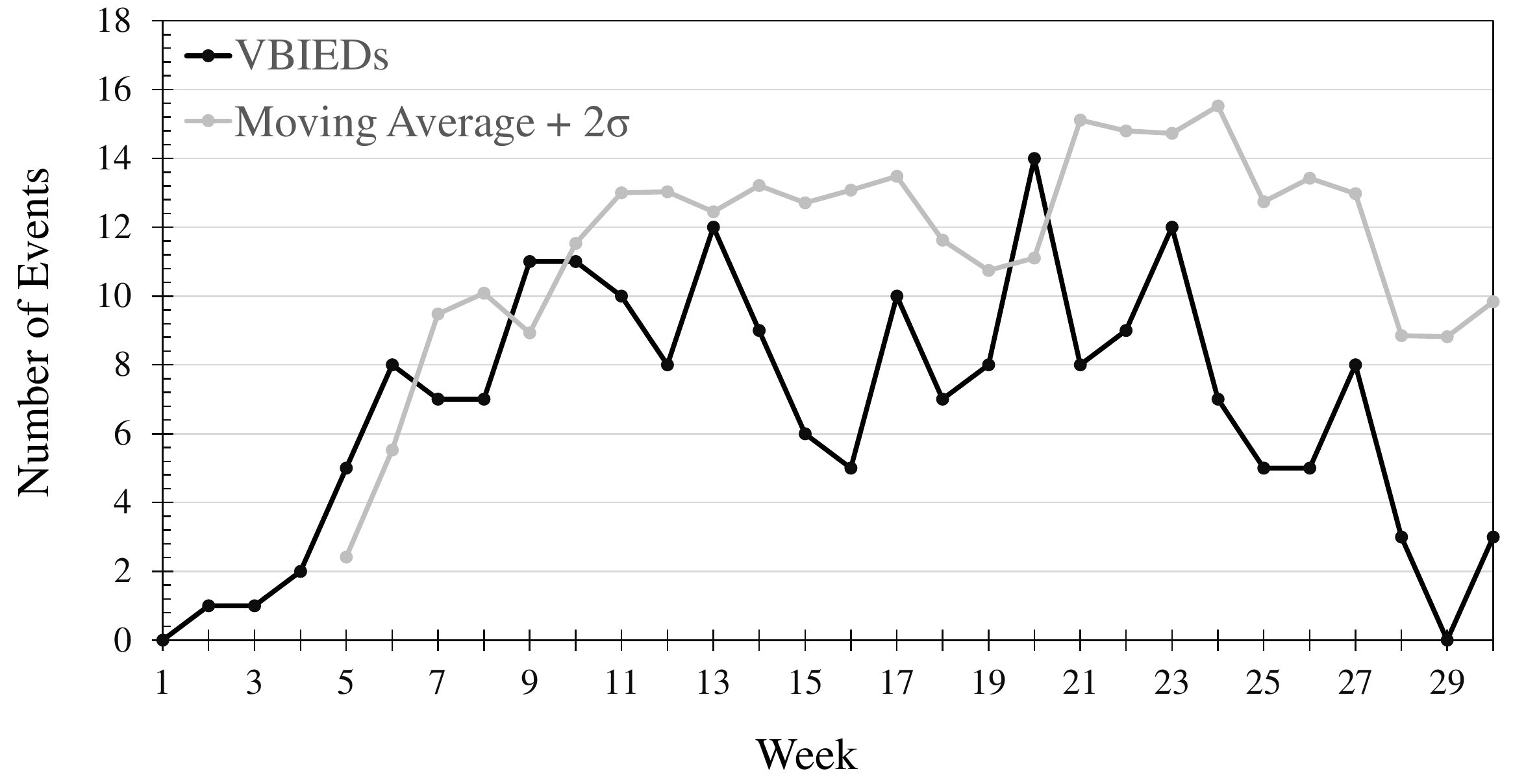}
\end{figure}

\noindent\textbf{Spikes in IED Incidents.}  Improvised explosive device (IED) incidents have been a common tactic used by Iraqi insurgents throughout the U.S.-led Operation Iraqi Freedom.  Though the IED comprises a smaller weapons system normally employed by local insurgent cells, spikes in such activity could be meaningful.  For instance, it may indicate action ordered by a strategic-level command that is being carried out on the city level, or it may indicate improved logistic support to provide local cells the necessary munitions to carry out such operations in larger numbers.  Such spikes only occur with a prior probability of $0.19$.  In Table~\ref{iedSpikeRules}, we show a rather strong precondition for such attacks that consist of infantry operations in Tikrit and a spike in executions.  In this Table, rule 4 states that infantry operations in Tikrit, when accompanied by a spike in executions, lead to a spike ($1\times\sigma$) (see figure \ref{fig:chart3}) in Iraq and Syria combined - with a probability of $1.0$ - much higher than the prior of the consequence.  With $\eavg$ of $0.97$ (based on a comparison with $1180$ other rules), this relationship appears highly causal ($\efrac=1.0, \emin=0.0$) although $40\%$ of $1\times \sigma$ IED spikes were not accounted for by this precondition.

\begin{figure}[h!]
\caption{Spikes in ISIS IEDs in Iraq and Syria per Week}
	\label{fig:chart3}
\includegraphics[width=8.5cm]{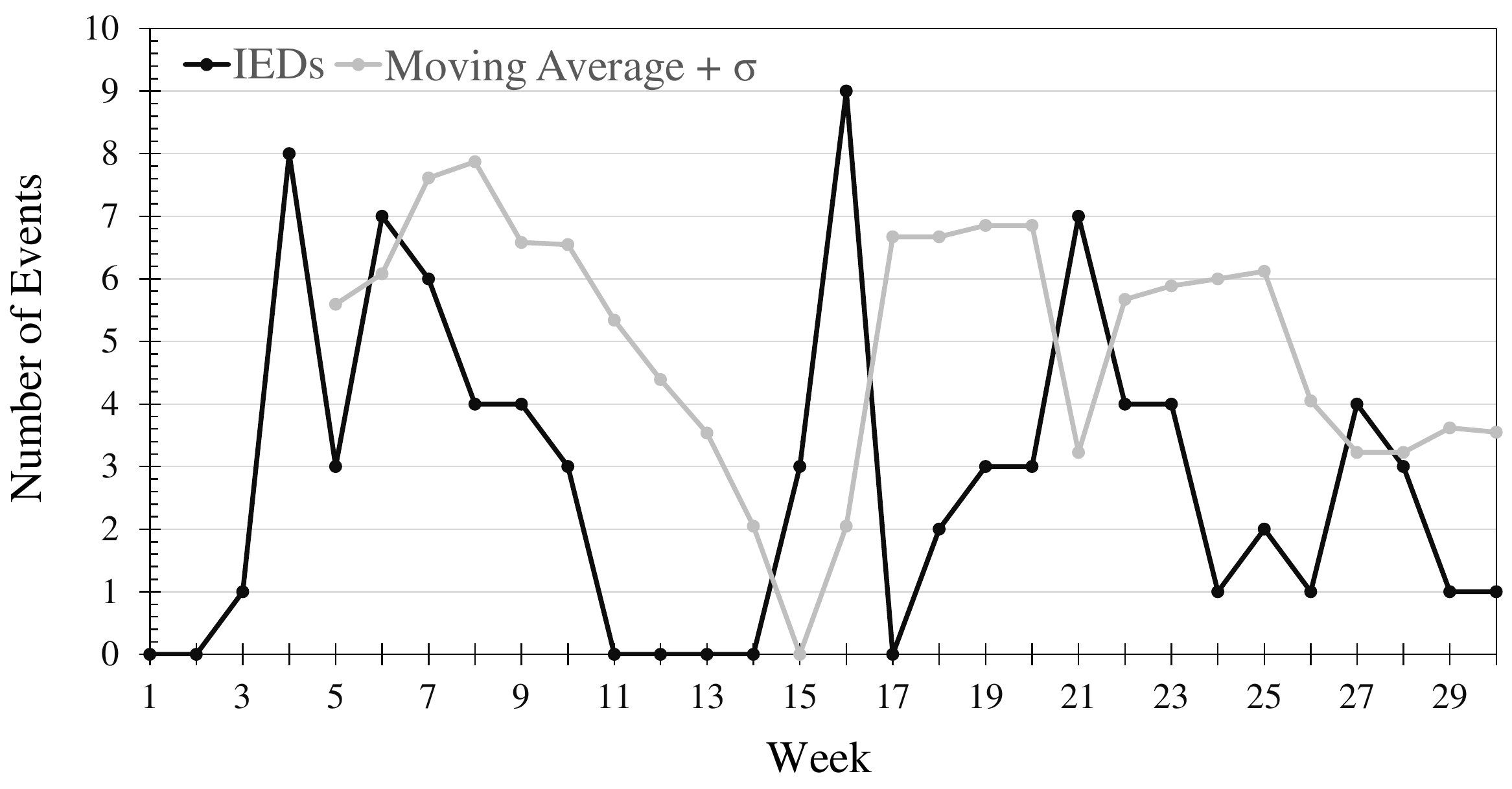}
\end{figure}

\begin{table}[h!]
	\caption{\textmd{Causal Rules for Spikes in IED Incidents in Iraq and Syria Combined}}
	\label{iedSpikeRules}
	\centering
	\renewcommand{\arraystretch}{1.5}
	
	\begin{tabular}{|c|c|c|c|c|} 
		\hline
		\textbf{No.} &{\bf Precondition} &  {\bf $\eavg$} & {\bf $p$} & {\bf $p^*$} \\ \hline\hline
		\textbf{4.}&$\textsf{armedAtk(}Tikrit\prn \wedge $ & $0.97$ & $1.00$ & $0.40$\\ 
		& $\textsf{executionSpike(}Total, 2\sigma\prn$&&&\\ \hline
		\hline
	\end{tabular}
\end{table}

\subsection{Relationships between Iraqi and Syrian Theaters}

ISIS has clearly leveraged itself as a force operating in both Iraq and Syria, and the identification of relationships between incidents in relation to the two theaters may indicate some sophisticated operational coordination on their part.  We have found some evidence of these cross-theater relationships with regard to suicide operations in Iraq (potentially affected by other events in Syria) and VBIED operations in Syria (potentially affected by other operations in Iraq).\smallskip\\

\noindent\textbf{VBIED Spikes in Syria.}  Table~\ref{vbiedSpikeSyria} shows our most causal rules whose consequence is a $1 \times \sigma$ spike (over the four-month moving average) (see figures \ref{fig:chart4}, \ref{fig:chart5}) in ISIS VBIED events in Syria.  Rule 5 provides a precondition of a spike in armed attacks in Iraq that includes a major spike in indirect fire activity while rule 6 has the same precondition but includes an additional VBIED event in Baghdad.  We believe that the spike in armed attacks indicates major ISIS operations in Iraq and the inclusion of indirect fire events also indicates that ISIS soldiers specializing in weapon systems such as mortars may also be concentrated in Iraqi operations.  Taken together, this may indicate a shift in ISIS resources toward Iraq - which may mean that operations have shifted away from Syria.  Hence, VBIED attacks, which once prepared, are less manpower-intensive, nevertheless provide a show-of-force in the secondary theater.  We also note that the probability of these rules ($1.00$) is significantly higher than the prior of these spikes ($0.19$) and that the causality value is also high for each of the $983$ related rules, as the probability either remains the same or increases when considering one of these preconditions (in other words, $\efrac = 1$ and $\emin=0$).\smallskip

\begin{figure}[h!]
\caption{Spikes in ISIS VBIEDs in Syria per Week}
	\label{fig:chart4}
\includegraphics[width=8.5cm]{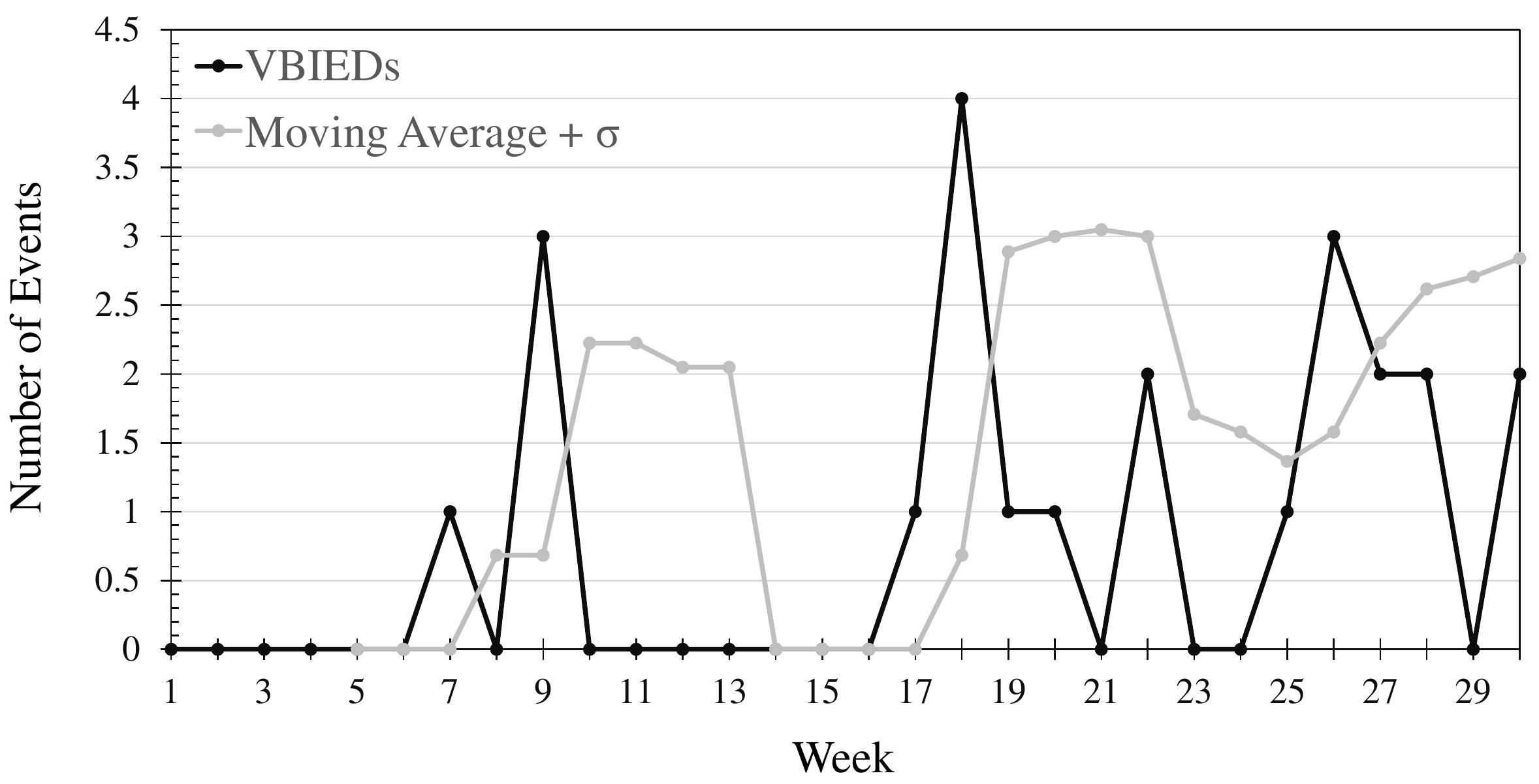}
\end{figure}

\begin{figure}[h!]
\caption{Comparison of ISIS Armed Attacks and VBIEDs in Syria per Week}
	\label{fig:chart5}
\includegraphics[width=8.5cm]{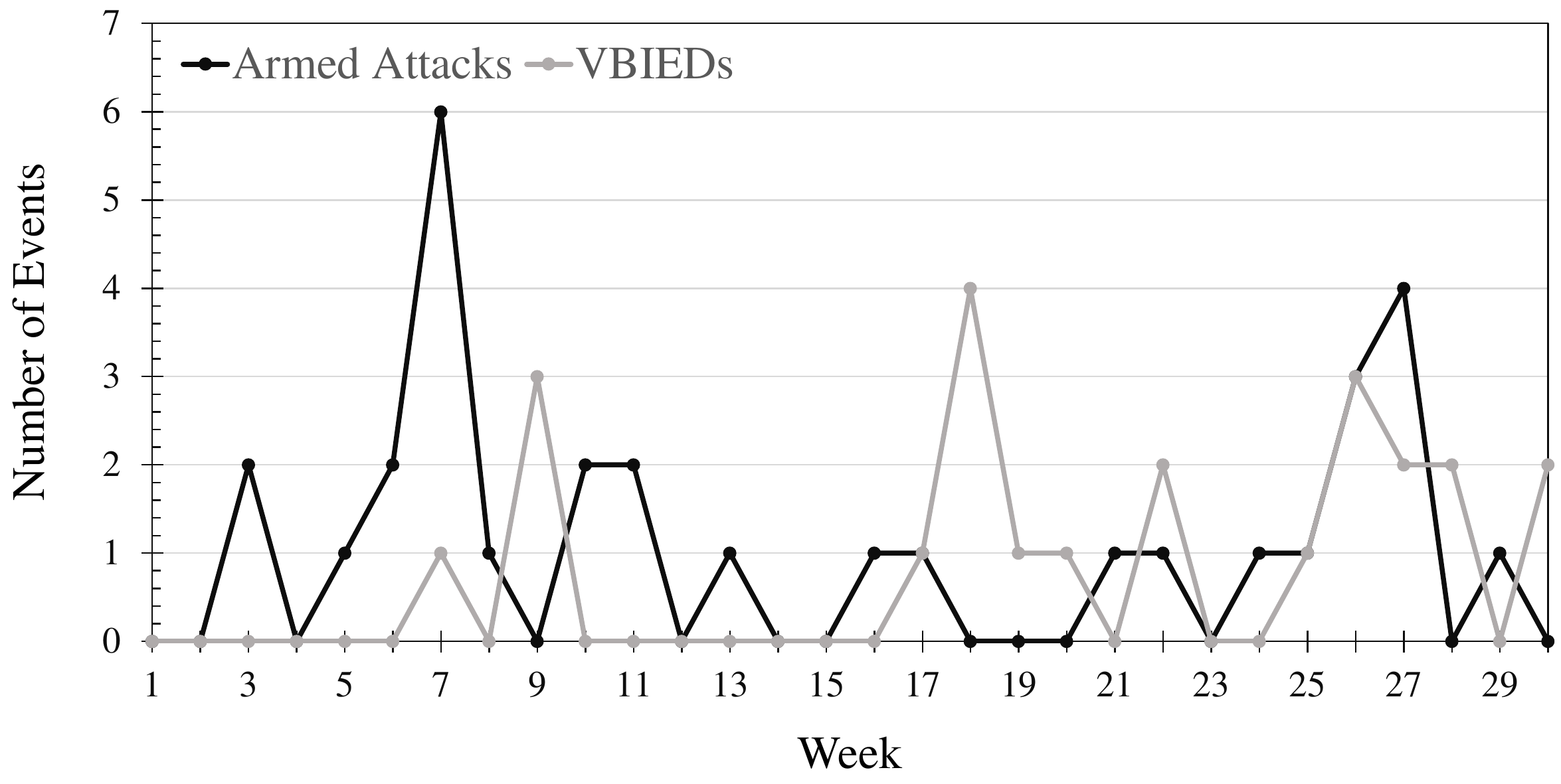}
\end{figure}

\begin{table}[h!]
	\caption{\textmd{Causal Rules for Spikes in VBIED Operations in Syria}}
	\label{vbiedSpikeSyria}
	\centering
	\renewcommand{\arraystretch}{1.5}
	
	\begin{tabular}{|c|c|c|c|c|} 
		\hline
		\textbf{No.} &{\bf Precondition} &  {\bf $\eavg$} & {\bf $p$} & {\bf $p^*$} \\ \hline\hline
		\textbf{5.}&$\textsf{armedAtkSpike(}Iraq,\sigma\prn \wedge $ & $0.92$ & $1.00$ & $0.20$\\ 
		& $\textsf{indirFireSpike(}Iraq,2\sigma\textsf{)}$&&&\\	\hline
		\textbf{6.}&$\textsf{armedAtkSpike(}Iraq,\sigma\prn \wedge $ & $0.92$ & $1.00$ & $0.20$\\ 
		& $\textsf{indirFireSpike(}Iraq,2\sigma\textsf{)}\wedge$&&&\\ 
		& $\textsf{VBIED(}Baghdad\textsf{)}$&&&\\  		\hline
					\hline
	\end{tabular}
\end{table}

\noindent\textbf{Operations Shifting to Syria.}  In rules 5 and 6 of Table~\ref{vbiedSpikeSyria}, we saw how increased operations in Iraq led to the use of VBIEDs in Syria -- perhaps due to a focus on more manpower-heavy operations in Iraq.  Interestingly, rule 7 shown in Table~\ref{idfSyria} indicates that less manpower-intensive operations in Iraq that have a terror component seem to have a causal relationship (based on $\eavg=0.97,\efrac=1.0, \emin=0.5$ found by comparing to $95$ related rules) with significant indirect fire operations in Syria (that occur with a prior probability of $0.08$).  As discussed earlier, indirect fire is normally used to ``prepare the battlefield'' for infantry operations, so this rule may be indicative of a shift in manpower toward Syria - and the use of a VBIED tactic in Baghdad (less manpower-intensive but very sensational) seems to always proceed a major ($2\times\sigma$ over the moving average) spike in indirect fire activity in Syria in our data (hence, a negative probability of $0.0$).\smallskip

%\begin{figure}
%\caption{Spikes in ISIS Indirect Fire Incidents in Syria per Week}
%\centering
%\includegraphics[width=8.5cm]{Chart6.pdf}
%	\label{fig:chart6}
%\end{figure}

\begin{table}[h!]
	\caption{\textmd{Causal Rules for Spikes in Indirect Fire Operations in Syria}}
	\label{idfSyria}
	\centering
	\renewcommand{\arraystretch}{1.5}
	
	\begin{tabular}{|c|c|c|c|c|} 
		\hline
		\textbf{No.} &{\bf Precondition} &  {\bf $\eavg$} & {\bf $p$} & {\bf $p^*$} \\ \hline\hline
		\textbf{7.}&$\textsf{IED(}Baghdad\prn\wedge $ & $0.97$ & $0.67$ & $0.00$\\ 
		& $\textsf{VBIED(}Ramadi\textsf{)}$&&&\\	\hline
					\hline
	\end{tabular}
\end{table}

\subsection{ISIS Activities Related to Opposition Air Strikes}

\begin{table}[h!]
	\caption{\textmd{Causal Rules for Spikes in Arrests (Iraq and Syria Combined)}}
	\label{respToSyr}
	\centering
	\renewcommand{\arraystretch}{1.5}
	
	\begin{tabular}{|c|c|c|c|c|} 
		\hline
		\textbf{No.} &{\bf Precondition} &  {\bf $\eavg$} & {\bf $p$} & {\bf $p^*$} \\ \hline\hline
		\textbf{8.}&$\textsf{airStrike(}SyrianGov,Damascus\prn $ & $0.91$ & $0.67$ & $0.00$\\ 
		\hline
					\hline
	\end{tabular}
\end{table}

\noindent\textbf{Reaction to Syrian Government Air Strikes.}  Rule 8, shown in Table~\ref{respToSyr}, tells us that a $2\times \sigma$ spike in arrests (see figure \ref{fig:chart7}) was \textit{always} ($p^*=0.0$) proceeded by an air strike conducted by the Syrian government.  The prior probability for such a spike in arrests is $0.08$.  Further, we found evidence of causality - the actions of the Syrian government raised the probability of each of $33$ related rules by at least $0.5$.  One potential reason to explain why arrests follow Syrian air operations is that unlike western nations, Syria generally lacks advanced technical intelligence-gathering capabilities to determine targets, and Syria likely relies on extensive human intelligence networks - especially within Iraq and Syria.  Successful targeting from the air by the Syrian government may then indicate that ISIS's counter-intelligence efforts (activities designed to locate spies within its ranks) may have failed and so the organization perhaps then decides to conduct massive arrests.\smallskip

\begin{figure}[h!]
\caption{Spikes in ISIS Arrests in Iraq and Syria per Week}
	\label{fig:chart7}
\includegraphics[width=8.5cm]{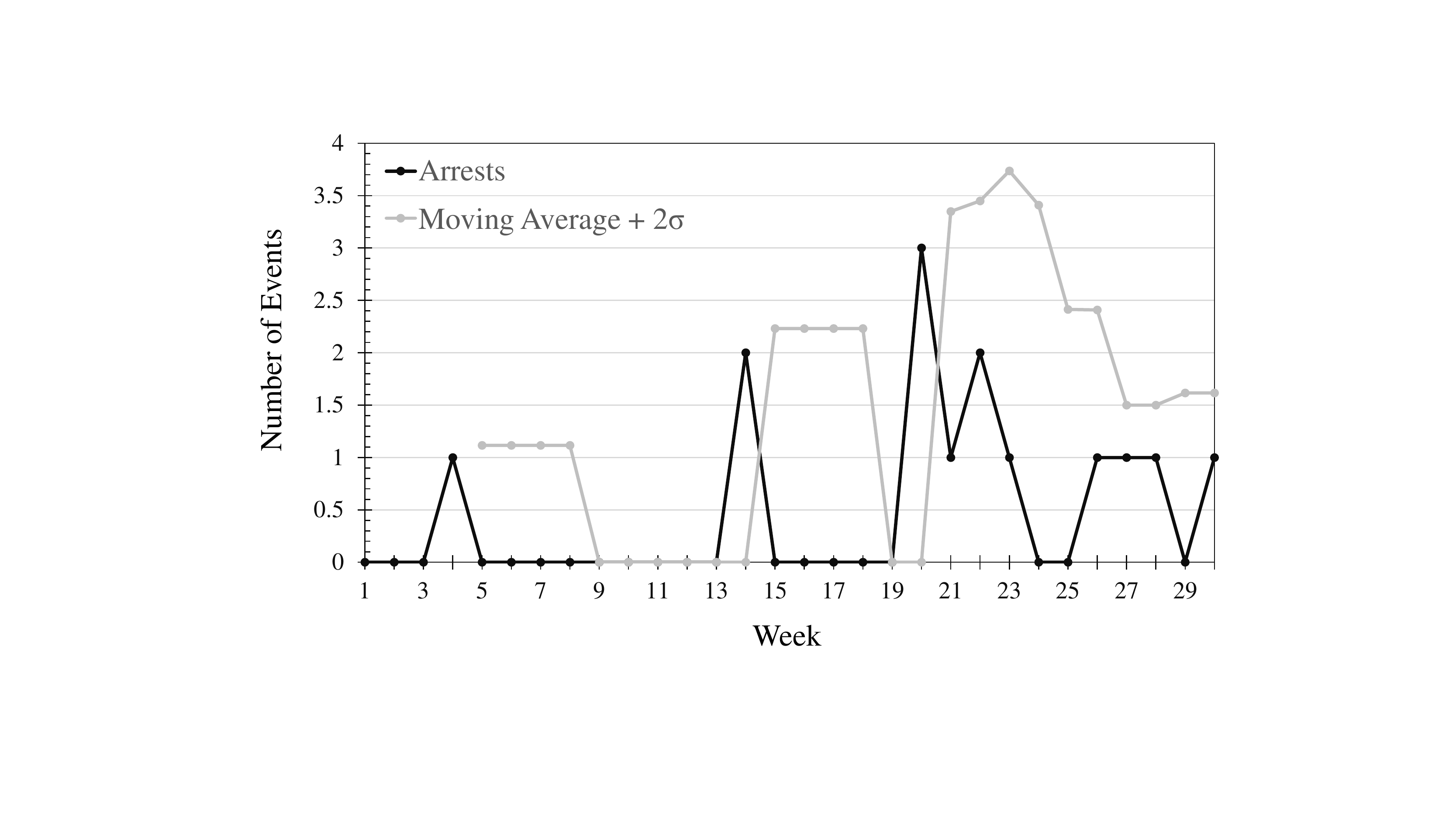}
\end{figure}

\begin{table}[h!]
	\caption{\textmd{Causal Rules for Major Spikes in Suicide Operations in Iraq}}
	\label{suicSpike}
	\centering
	\renewcommand{\arraystretch}{1.5}
	
	\begin{tabular}{|c|c|c|c|c|} 
		\hline
		\textbf{No.} &{\bf Precondition} &  {\bf $\eavg$} & {\bf $p$} & {\bf $p^*$} \\ \hline\hline
		\textbf{9.}&$\textsf{airStrike(}IraqGovt,Baiji\prn\wedge $ & $0.71$ & $0.67$ & $0.60$\\ 
		&$\textsf{armedAtk(}Balad\prn\wedge $ &  &  & \\ 
		&$\textsf{VBIED(}Baghdad\prn $ &  &  & \\ 
				\hline
					\hline
	\end{tabular}
\end{table}

\noindent\textbf{Reaction to Iraqi Government Air Strikes.}  In Table~\ref{suicSpike}, rule 9 tells us that a $2\times \sigma$ spike (see figure \ref{fig:chart8}) in suicide operations in Iraq is related to a precondition involving Iraqi aerial operations in Baiji, ISIS infantry operations in Balad, and a VBIED in Baghdad.  The rule shows that these spikes in suicide operations are $3.5$ times more likely with this precondition. Though $\eavg$ is lower than some of the other rules presented thus far, it has a value for $\emin$ of $0.3$ - the minimum increase in probability afforded to any of the $82$ related rules to which we add the precondition of rule 9.  The precondition indicates that ISIS may have recently expended a VBIED (expensive in terms of equipment) while it has ongoing infantry operations in Balad (expensive in terms of personnel) - hence, it may seek a more economical attack (in terms of both manpower and equipment) that still has a significant terror component - and it would appear that a suicide attack provides a viable option to respond to such air strikes.

\begin{figure}[h!]
\caption{Spikes in ISIS Suicide Operations in Iraq per Week}
	\label{fig:chart8}
\includegraphics[width=8.5cm]{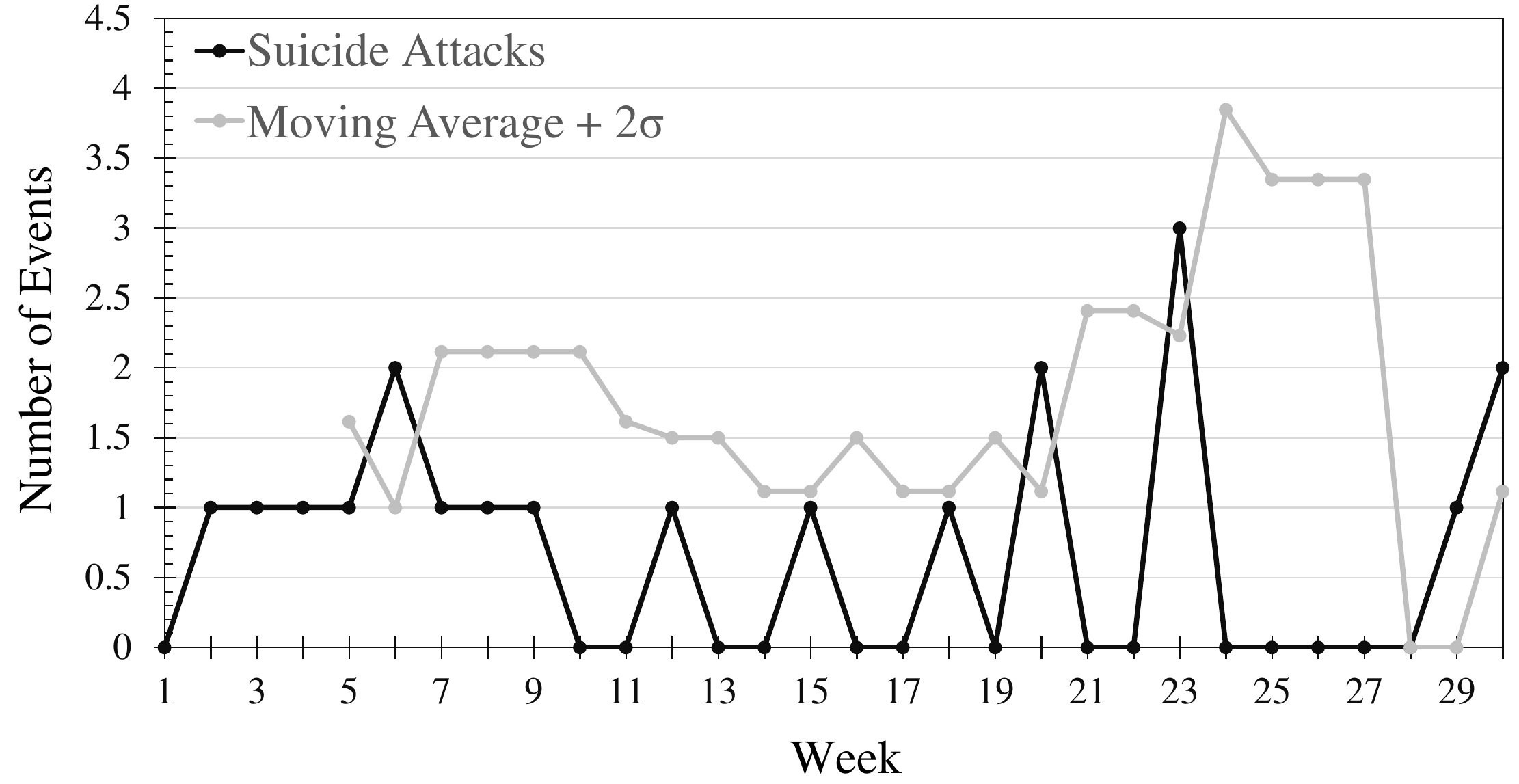}
\end{figure}

\begin{table}[h!]
	\caption{\textmd{Causal Rules for Major Spikes in IED Operations in Iraq}}
	\label{izIedToAir}
	\centering
	\renewcommand{\arraystretch}{1.5}
	
	\begin{tabular}{|c|c|c|c|c|} 
		\hline
		\textbf{No.} &{\bf Precondition} &  {\bf $\eavg$} & {\bf $p$} & {\bf $p^*$} \\ \hline\hline
		\textbf{10.}&$\textsf{airStrike(}Coalition,Mosul\prn\wedge $ & $0.97$ & $0.67$ & $0.33$\\ 
		&$\textsf{armedAtk(}Fallujah\prn $ &  &  & \\ 
				\hline
					\hline
	\end{tabular}
\end{table}

\begin{table}[h!]
	\caption{\textmd{Causal Rules for Major Spikes in IED Operations in Syria}}
	\label{syIedToAir}
	\centering
	\renewcommand{\arraystretch}{1.5}
	
	\begin{tabular}{|c|c|c|c|c|} 
		\hline
		\textbf{No.} &{\bf Precondition} &  {\bf $\eavg$} & {\bf $p$} & {\bf $p^*$} \\ \hline\hline
		\textbf{11.}&$\textsf{airStrike(}Coalition,Mosul\prn\wedge $ & $0.79$ & $0.67$ & $0.33$\\ 
		&$\textsf{armedAtk(}Ramadi\prn\wedge $ &  &  & \\ 
		&$\textsf{armedAtkSpike(}Syria, \sigma\prn $ &  &  & \\ 
				\hline
					\hline
	\end{tabular}
\end{table}

\noindent\textbf{Reaction to Air Strikes by the U.S.-led coalition.}  Rules 10 ($\efrac=1.0, \emin=0.5, 562$ related rules) and 11 ($\efrac=1.0, \emin=0.38, 81$ related rules) - shown in Tables~\ref{izIedToAir} and \ref{syIedToAir} - illustrate two different outcomes from coalition air operations in Mosul.  In both cases, ISIS has active operations in the Al-Anbar province -- both Ramadi and Fallujah have similar demographics.  Hence, the major difference is that the precondition of rule 11 also includes significant infantry operations in Syria.  So, even though both rules result in an increase in IED activity ($2 \times \sigma$ above average, a spike that in both cases occurs with a prior probability of $0.12$) (see figures \ref{fig:chart9}, \ref{fig:chart10}) - the increase occurs in Iraq for rule 10 and Syria for rule 11.  It may be the case that ISIS is increasing IED activity in response to coalition air strikes in the location of their main effort. Further, the relatively low negative probability ($0.33$) along with relatively high values for $\emin$ may indicate that coalition air strikes in Mosul are likely viewed as important factors contributing to ISIS decisions to increase IED activity.  We may also note that the use of IED activity in the aftermath of coalition air strikes could also be due to the size of the weapon.  An IED can be fairly small, easily disassembled and stored in an innocuous location - hence, it is weapon system that the coalition cannot sense and target from the air.

\begin{figure}[h!]
\caption{Spikes in ISIS IEDs in Iraq per Week}
	\label{fig:chart9}
\includegraphics[width=8.5cm]{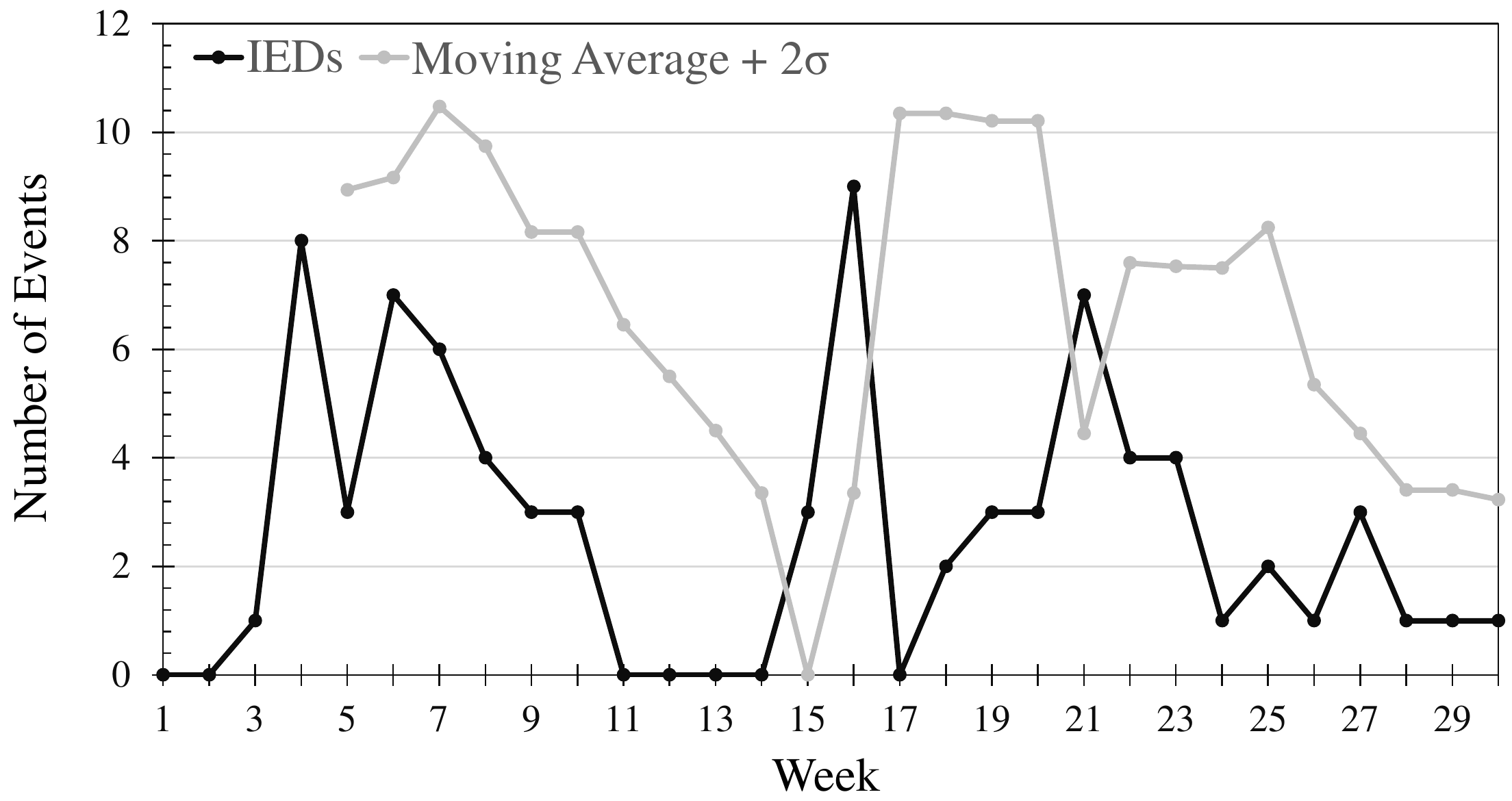}
\end{figure}

\begin{figure}[h!]
\caption{Spikes in ISIS IEDs in Syria per Week}
	\label{fig:chart10}
\includegraphics[width=8.5cm]{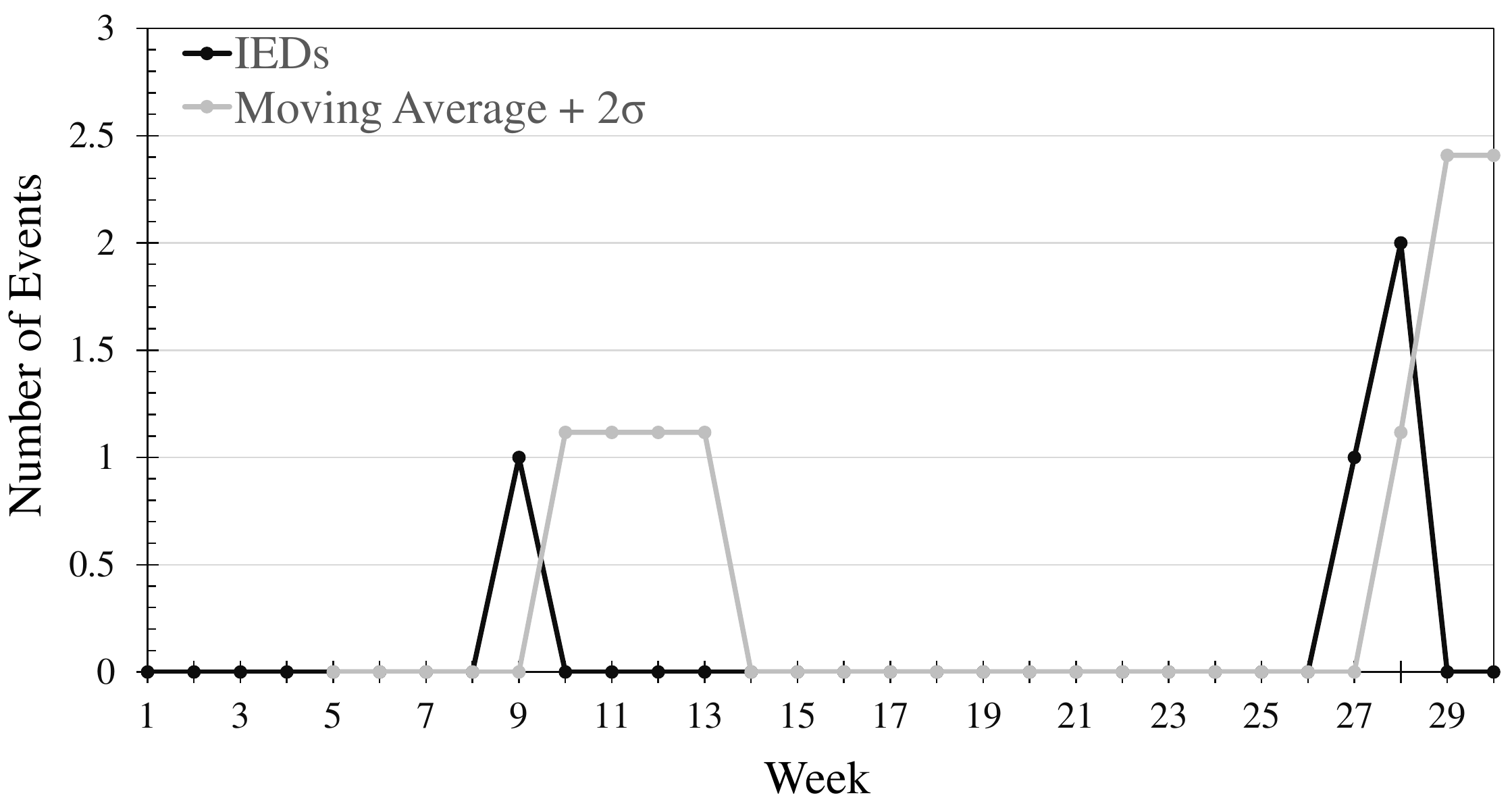}
\end{figure}

%% file: related_work.tex
\section{Related Work}
\label{rwSec}

To the best of our knowledge, this paper represents the first purely data-driven study of ISIS, and it is also the first to combine the causality framework of \cite{kb09} with APT-Logic.  However, there has been a wealth of research conducted where rule-based systems have been applied to terrorist and insurgent groups, as well as other work on modeling insurgent actions.  Here, we review this related work and also discuss other approaches to causal reasoning.\smallskip\\

\noindent\textbf{Rule-based systems.}  Previously, there has been a variety of research on modeling the actions of terrorist and insurgent groups.  Perhaps most well-known is the Stochastic Opponent Modeling Agents (SOMA)~\cite{soma1,soma}, which was developed with the goal of better understanding different cultural groups along with their behaviors.  This is also based on a rule-based framework known as action-probabilistic (AP) rules~\cite{apPgms07}, which is a predecessor of APT-logic used in this paper~\cite{apt11,apt12}.  The SOMA system produces a probabilistic logic representation for understanding group behaviors and reasoning about the types of actions a group may take. In recent years, this system has been used to better understand terror groups like Hamas~\cite{soma} and Hezbollah~\cite{soma1} through probabilistic rules.  However, it can also be used with cultural, religious, and political groups - as it has been demonstrated with other groups in the Minorities at Risk Organizational Behavior (MAROB) dataset~\cite{marob}.  The work on SOMA was followed by an adoption of APT-Logic~\cite{apt11,apt12}, which extended AP-rules with a temporal component.  APT logic was also applied to the groups of the MAROB dataset as well as Shi'ite Iraqi insurgent groups circa 2007 (present during the American-led Operation Iraqi Freedom)~\cite{apt12}~\footnote{ISIS is a Sunni group and did not exist in its present form in 2007.}.  Perhaps most significantly, APT-Logic was applied to another dataset for the study of the terror groups Lashkar-e-Taiba (LeT) and Indian Mujahideen (IM) in two comprehensive volumes that discuss the policy implications of the behavioral rules learned in these cases~\cite{let13,im13}.  The main similarity between the previous work leveraging SOMA, AP-rules, and/or APT-logic and this paper is that both leverage probabilistic rule learning.  However, all of the aforementioned rule-learning approaches only discover correlations, while this work is focused on rules that are not only of high probability, but also have a likely casual relationship.  We also note that none of this previous work studies ISIS but rather other terrorist and insurgent groups.\smallskip\\

\noindent\textbf{Analysis of insurgent military tactics.}  We also note that the previous rule-learning approaches to understanding terrorist and insurgent behavior have primarily focused on analyzing political events and how they affect the actions of groups such as LeT and IM.  However, with the exception of some of the work on APT-Logic which studies Shi'ite Iraqi insurgents~\cite{apt11}, the aforementioned work is generally not focused on tactical military operations - unlike this paper.  However, there have been other techniques introduced in the literature that have been designed to better account for ground action.  Previously, a statistical-based approach leveraged leaked classified data~\cite{Zammit-Mangion31072012}~\footnote{Resulting from WikiLeaks in 2010.} to predict trends of violent activity during the American-led Operation Enduring Freedom in Afghanistan.  However, this work was primarily focused on prediction and not identifying causal relationships of interest to analysis as in this work.  We also note that this work does not rely on the use of leaked classified data, but rather on open-source information.  Another interesting approach is the model-based approach of \cite{gian14}, in which subject matter experts create models of insurgent behavior using a combination of fuzzy cognitive maps and complex networks to run simulations and study ``what-if'' scenarios.  However, unlike  this paper, the work of \cite{gian14} is not as data-driven an approach.\smallskip\\

\noindent\textbf{Causal reasoning.}  The comparison of rules by $\eavg$ as a measure of causality was first introduced in~\cite{kb09} and further studied in \cite{kb11}.  It draws on the philosophical ideas of \cite{suppes70}.  However, this work has primarily looked at preconditions as single atomic propositions.  Further, it did not explore the issue of efficiency.  As we look to identify rules whose precondition consists of more than one atomic proposition, we introduce a rule-learning approach.  Further, moving beyond previous rule-learning approaches such as that introduced in \cite{apPgms07} and \cite{apt11}, we show practical techniques to improve efficiencies of such algorithms.  Further, we also improve upon the efficiency of computing $\eavg$ - a practical issue not discussed in \cite{kb09,kb11}.  We note that neither this previous work on causality nor APT-logic makes independence assumptions.  This differs substantially from earlier work on causality such as \cite{cau00}, which relies on graphical models that make relatively strong independence assumptions.  Our goal was to avoid such structures in order to provide a more purely data-driven approach.  A key aspect of this work as opposed to previous studies on causality is that we leverage the intuitions of APT logic and AP rules, in which the rule-learning algorithm searches for combinations of atomic propositions that are related to a given consequence.

%% file: conclusion.tex
\section{Conclusion}
In this paper, we conducted a data-driven study on the insurgent group ISIS using a combination of APT-logic, rule learning, and causal reasoning to identify cause-and-effect behavior rules concerning the group's actions.  We believe our approach is of significant utility for both military decision making and the creation of policy.  In the future, we look to extend this work in several ways: first, we look to create non-ground rules that generalize some of the preconditions further.  This would allow us to understand the circumstances in which ISIS conducts general operations (as opposed to operations specific to a given city or to a geographic area).  We also look to study more complex temporal relationships -- possibly using a more fine-grain resolution for time and studying rules where the cause is followed by the effect in more than one unit of time.  We also look to leverage additional variables about the environment, including data about weather, information (including social media operations), and the political situation to find more interesting relationships.